\documentclass[%
 reprint,
 amsmath,amssymb,
 aps,
pra,
]{revtex4-1}

\usepackage{amsmath}
\usepackage{amssymb}
\usepackage{mathrsfs}
\usepackage{braket}
\usepackage{mathtools}
\usepackage{color}
\newcommand{\1}{\mbox{1}\hspace{-0.25em}\mbox{l}}

\usepackage{amsthm}
\newtheorem{theorem}{Theorem}

\newtheorem{lemma}{Lemma}

\newtheorem{proposition}{Proposition}
\begin{document}


\title{
Operational characterization of incompatibility of quantum channels \\
with quantum state discrimination
}

\author{Junki Mori}
 \email{mori.junki.36u@st.kyoto-u.ac.jp}
\affiliation{%
 Department of Nuclear Engineering, Kyoto University, 6158540 Kyoto, Japan
}%

\begin{abstract}
  A collection of quantum channels is called incompatible if they cannot be obtained as 
  marginals from a single channel. No-cloning theorem is the most prominent instance of incompatibility of quantum channels. We show that every collection of incompatible channels can be more useful than compatible ones as preprocessings in a state discrimination task with multiple ensembles of states. This is done by showing that
  the robustness of channel incompatibility which is a measure for incompatibility of channels exactly quantifies the maximum advantage in the state discrimination.
  We also show that incompatibility of quantum measurement and channel has a similar 
  operational interpretation. Finally, we demonstrate that our result with respect 
  to channel incompatibility includes all other kinds of incompatibility as special cases.    
\end{abstract}

\pacs{Valid PACS appear here}
\maketitle


\section{Introduction}

Incompatibility is one of the fundamental features in quantum theory \cite{Teiko2016}.
It is well known that there exist quantum measurements which cannot be implemented simultaneously, and such measurements are called incompatible.
As only incompatible measurements can lead to the violation of a Bell's inequality \cite{Wolf}, 
incompatibility of measurements is not an obstacle but rather an advantage.
Therefore, it is natural to consider incompatibility in the resource theoretic perspective \cite{Chitambar}. 
In fact, as with the robustness of entanglement, which is one of the 
measures for entanglement \cite{Vidal,Steiner}, the robustness of measurement incompatibility has been studied \cite{Erkka,RoopeOne,SkrzypczkAll,RoopeQuanti}. 
In particular, Ref.~\cite{SkrzypczkAll,RoopeQuanti,CarmeliQuan} showed that every collection of 
incompatible measurements can give an advantage over all compatible ones in a quantum 
state discrimination task with multiple ensembles of states. Furthermore, in Ref.~\cite{SkrzypczkAll,RoopeQuanti} it was also 
shown that the maximum advantage is exactly quantified by the robustness of measurement 
incompatibility. These results are among the recent results in the 
resource theories of states \cite{PinaniNece,Napoli,PianiRob,Bae,TakagiOpe,TakagiGene}, measurements \cite{TakagiGene,SkrzypczkRob,Oszmaniec2019operational}, and channels \cite{TakagiGene}, which showed that the robustness measures have an operational meaning as advantages in some discrimination tasks.

Incompatibility can be defined not only for measurements but also for channels \cite{Teiko2016,Teiko2017}.
A collection of channels is said to be incompatible if they cannot be obtained from a single channel.
For example, No-cloning theorem results from the incompatibility of the two identity 
channels.
The robustness of channel incompatibility can also be defined \cite{Erkka}. 
In this paper, we show that the robustness of channel incompatibility exactly quantifies 
the advantage provided by an incompatible collection of channels in a particular state discrimination task.
Our result gives an operational interpretation to incompatibility of quantum channels as well as incompatibility of quantum measurements in terms of state discrimination tasks.
Moreover, our result can be applied to measurement incompatibility via quantum-to-classical channels. We also show that a similar relation holds for incompatibility of quantum measurement and channel. 

This paper is organized as follows. In Sec.~\ref{pre}, we review briefly quantum incompatibility and introduce the state discrimination task with multiple ensembles of states originally introduced in Ref.~\cite{ballester2008state}.  We state our main theorems in Sec.~\ref{main}. In Sec.~\ref{conc}, we give a summary.

\section{preliminaries}
\label{pre}

\subsection{quantum incompatibility}
 In this section, we introduce the concept of incompatibility of quantum measurements, quantum channels, and a pair of a quantum measurement and a quantum channel.

 Let $\mathcal{H}$ and $\mathcal{K}$ be finite-dimensional Hilbert spaces.
 A measurement on a physical system on $\mathcal{H}$ with a finite outcome set $\Omega=\{1,\ldots ,o\}$ is described as a positive operator valued measure (POVM) 
 $\mathbb{M}=\{M_i\}_{i=1}^{o}$, where each element $M_i$ is a positive-semidefinite operator on $\mathcal{H}$ satisfying 
 $\sum_{i}M_i = \1$. We denote the set of POVMs on $\mathcal{H}$ with an outcome set $\Omega$ by $\mathfrak{P}(\Omega,\mathcal{H})$.
 A collection of $n$ measurements $\{\mathbb{M}_x\}_x=\{M_{i|x}\}_{i,x}$ is said to be compatible if there exist a joint measurement $\mathbb{G}=\{G_{\lambda}\}$ and conditional probability distributions $p(i|x,\lambda)$ such that
 \begin{align}
	 M_{i|x} = \sum_{\lambda} p(i|x,\lambda) G_{\lambda} \quad \forall \ i,x.
 \end{align} 
 Otherwise, the collection is called incompatible.
 We denote the set of collections of $n$ POVMs which are elements of $\mathfrak{P}(\Omega,\mathcal{H})$ by $\mathfrak{P}^{n}(\Omega,\mathcal{H})$, and the subset of compatible ones by $\mathfrak{P}^{n}_{\rm com}(\Omega,\mathcal{H})$.

 A quantum channel which transforms a state on $\mathcal{H}$ into a state on $\mathcal{K}$ is a completely positive 
 trace-preserving map $\Lambda \colon \mathcal{S}(\mathcal{H}) \rightarrow \mathcal{S}(\mathcal{K})$, where $\mathcal{S}(\mathcal{H})$ (respectively $\mathcal{S}(\mathcal{K})$) is the set of states on $\mathcal{H}$ (respectively $\mathcal{K}$). 
 We denote the set of channels $\Lambda \colon \mathcal{S}(\mathcal{H}) \rightarrow \mathcal{S}(\mathcal{K})$ by 
 $\mathfrak{C}(\mathcal{H}, \mathcal{K})$. A collection of $n$ channels $\{\Lambda_x\}_{x}$ ($\Lambda_x\in\mathfrak{C}(\mathcal{H}, \mathcal{K}_x)$) is said to be compatible if there exists a channel 
 $\Lambda \in \mathfrak{C}(\mathcal{H}, \mathcal{K}_1 \otimes \cdots \otimes \mathcal{K}_n)$ such that
 \begin{align}
 \label{cha}
	 \Lambda_x = {\rm Tr}_{\overline{\mathcal{K}_x}}[\Lambda(\ \cdot\ )] \quad \forall \ x,
 \end{align}
 where ${\rm Tr}_{\overline{\mathcal{K}_x}}$ means taking the partial trace except for the Hilbert space $\mathcal{K}_x$.
 Otherwise, the collection is called incompatible.
 Throughout the paper we restrict ourselves to collections of channels that have the same output space
 for simplicity. We denote the set of collections of $n$ channels which are elements of $\mathfrak{C}(\mathcal{H},\mathcal{K})$ 
 by $\mathfrak{C}^{n}(\mathcal{H},\mathcal{K})$, and the subset of compatible ones by 
 $\mathfrak{C}^{n}_{\rm com}(\mathcal{H},\mathcal{K})$.
 
 Finally, we introduce the concept of incompatibility of quantum measurement and channel.
 A measurement process is mathematically described by an instrument 
 $\mathcal{I}=\{\mathcal{I}_i\}_{i\in\Omega}$, where each $\mathcal{I}_i$ is a completely 
 positive trace-non-increasing map and $\sum_{i}\mathcal{I}_i$ is a channel.
 The probability of getting the outcome $i$ given the initial state $\rho$ is $p_{\rho}(i)={\rm Tr}[\mathcal{I}_i(\rho)]$ and the conditioned state is $\mathcal{I}_i(\rho)/p_{\rho}(i)$.
 We denote the set of instruments that have the outcome set $\Omega$ and transform a state 
 on $\mathcal{H}$ into a state on $\mathcal{K}$ by $\mathfrak{I}(\Omega,\mathcal{H},\mathcal{K})$. A pair of a measurement and a channel 
 $\{\mathbb{M},\Lambda\} \in \mathfrak{P}(\Omega,\mathcal{H})\times\mathfrak{C}(\mathcal{H},\mathcal{K})$ is compatible if there exists an instrument 
 $\mathcal{I} \in \mathfrak{I}(\Omega,\mathcal{H},\mathcal{K})$ such that 
 \begin{align}
 	M_i=\mathcal{I}_{i}^{*}(\1),\quad \Lambda(\rho)=\sum_{i}\mathcal{I}_i(\rho) 
 	\quad \forall \ i,\rho,
 \end{align}
 where $\mathcal{I}_{i}^{*}$ is the Heisenberg picture of $\mathcal{I}_i$. 
 Otherwise, the pair is called incompatible.
 
\subsection{state discrimination}
 In this paper, we consider the following two-party state discrimination task as in Ref.~\cite{SkrzypczkAll,RoopeQuanti,CarmeliQuan}, originally introduced in \cite{ballester2008state}. Bob can prepare $n$ defferent ensembles $\{\mathcal{E}_x\}_x$ labeled by $x$, where each ensemble consists of $o$ quantum states $\mathcal{E}_x=\{p(i|x), \rho_{i|x}\}_{i=1}^{o}$. First, Bob chooses a label $x$ with probability $p(x)$. According to the label $x$, Bob picks one of the states $\rho_{i|x}$ from the ensemble $\mathcal{E}_x$ with probability $p(i|x)$ and sends Alice his choice $x$ and the state $\rho_{i|x}$.  After receiving the state and the value of $x$, Alice aims at correctly guessing $i$. In this setting, since Alice knows which ensemble was chosen, she can change her operation performed on the state to guess correct $i$ according to 
 the chosen ensemble. We then focus on the property of the collection of operations such as measurements and state transformations that Alice performs on the states to discriminate different ensembles $\{\mathcal{E}_x\}$. For example, we consider two qubit state ensembles $\mathcal{E}_1$ and $\mathcal{E}_2$. We fix an orthonormal basis $\{\ket{0},\ket{1}\}$ and suppose that $\mathcal{E}_1$ consists of two pure states $\ket{0}$, $\ket{1}$ chosen uniformly at random and $\mathcal{E}_2$ consists of two pure states $\ket{+}$, $\ket{-}$ chosen uniformly at random where $\ket{\pm}=\frac{1}{\sqrt{2}}(\ket{0}\pm\ket{1})$, as is the case in the BB84 protocol in quantum key distribution. If Alice performs the two projective measurements $\mathbb{M}_1=\{|0\rangle\langle 0|,|1\rangle\langle 1|\}$ and $\mathbb{M}_2=\{|+\rangle\langle +|,|-\rangle\langle -|\}$ on the states of the ensembles $\mathcal{E}_1$ and $\mathcal{E}_2$, respectively, she can perfectly discriminate $\mathcal{E}_1$ and $\mathcal{E}_2$. We note that the two measurements $\mathbb{M}_1$ and $\mathbb{M}_2$ are incompatible.
    
 In the case where Alice has access to only a fixed collection of measurements $\{\mathbb{M}_x\}_x \in \mathfrak{P}^{n}(\Omega,\mathcal{H})$, she performs a measurement $\mathbb{M}_x$ after receiving the value of $x$ as is the case in the above example. In this setting, the average probability that she successfully obtains the correct guess is given by
 \begin{align}
     P_{\rm succ}(\mathscr{A}, \{\mathbb{M}_x\})=\sum_{x,i}p(x)p(i|x){\rm Tr}[\rho_{i|x}M_{i|x}],
 \end{align}
 where $\mathscr{A}=\{p(x), \mathcal{E}_x\}$. Ref.~\cite{SkrzypczkAll,RoopeQuanti} showed that the 
 maximum advantage of $\{\mathbb{M}_x\}$ over compatible measurements is quantified as
 \begin{align}
  \label{measurement}
    \max_{\mathscr{A}}
    \dfrac{
    	P_{\rm succ}(\mathscr{A}, \{\mathbb{M}_x\})}
    {\displaystyle\max_{\{\mathbb{N}_x\} \in \mathfrak{P}^{n}_{\rm com}(\Omega,\mathcal{H})}
    	P_{\rm succ}(\mathscr{A}, \{\mathbb{N}_x\})} = 1+R_M(\{\mathbb{M}_x\}),
 \end{align}
 where $R_M(\{\mathbb{M}_x\})$ is the robustness of measurement incompatibility, 
 which is a geometric quantifier for incompatibility of measurements defined as the 
 minimum amount of noise which has to be added to $\{\mathbb{M}_x\}$ to make them compatible. It is formally defined as   
 \begin{align}
 	R_M(\{\mathbb{M}_x\}) = &\min{s} \notag \\[2pt]
 	                      & \ \ {\rm s.t.}\  
 	                      \left\{\frac{\mathbb{M}_x + s\mathbb{M}_x'}{1+s}\right\}
 	                      \in\mathfrak{P}^{n}_{\rm com}(\Omega,\mathcal{H}), \notag \\[5pt] 
 	                      &\qquad \ \ \{\mathbb{M}_x'\}\in\mathfrak{P}^{n}(\Omega,\mathcal{H}),
 \end{align}
where the minimization is over $s\geq0$ and $\{\mathbb{M}_x'\}$.

\section{main results}
\label{main}
\subsection{robustness of channel incompatibility}
 In this section, we prove that a similar relation to Eq.~(\ref{measurement}) holds for 
 channels by using conic programming as with Ref.~\cite{RoopeQuanti,TakagiGene}. 
 In particular, we utilize the method of Ref.~\cite{TakagiGene} with respect to the resource theory 
 of a single channel.

 For any collection of channels 
 $\{\Lambda_x\} \in \mathfrak{C}^{n}(\mathcal{H},\mathcal{K})$, the robustness of channel incompatibility $R_C(\{\Lambda_x\})$ is defined in the same way as the 
 robustness of measurement incompatibility, that is, 
 \begin{align}
   \label{robust}
    R_C(\{\Lambda_x\}) = &\min{s} \notag \\[2pt]
                         & \ {\rm s.t.}\  
                         \left\{\frac{\Lambda_x + s\Lambda'_x}{1+s}\right\}
                         \in \mathfrak{C}^{n}_{\rm com}(\mathcal{H},\mathcal{K}), 
                         \notag \\[5pt] 
                         &\qquad \ \{\Lambda_x'\}\in
                         \mathfrak{C}^{n}(\mathcal{H},\mathcal{K}), 
 \end{align}
 where the minimization is over $s\geq0$ and $\{\Lambda_x'\}$.
 We note that the channel $\Phi_x=\frac{\Lambda_x + s\Lambda'_x}{1+s}$ is defined as
 \begin{align}
 	\Phi_x(\rho):=\frac{\Lambda_x(\rho) + s\Lambda'_x(\rho)}{1+s}\quad \forall\rho\in\mathcal{S}(\mathcal{H}).
 \end{align}
 In order to make the 
 robustness conected to the state discrimination, we consider the dual 
 problem equivalent to the optimization problem Eq.~(\ref{robust}). For that purpose, we utilize the Choi-Jamiolkowski isomorphism. $\mathfrak{C}(\mathcal{H}, \mathcal{K})$ is isomorphic to a subset  
 $\mathcal{J}(\mathcal{K}\otimes\mathcal{H}) \subset \mathcal{S}(\mathcal{K}\otimes\mathcal{H})$ with its elements satisfying ${\rm Tr}_{\mathcal{K}}[J]=\1_{\mathcal{H}}/d$, where $d$ is the dimension of $\mathcal{H}$. This isomorphism is defined by 
 the mapping $\Lambda \mapsto J_{\Lambda}=(\Lambda\otimes{\rm id})(|\Psi_+\rangle\langle\Psi_+|)$, where $\ket{\Psi_+}\in\mathcal{H}\otimes\mathcal{H}$ is the maximally entangled state
 and ${\rm id}$ denotes the identity channel. Let $\mathcal{J}^{n}(\mathcal{K}\otimes\mathcal{H})$ be the set of collections of $n$ Choi matrices and $\mathcal{J}^{n}_{\rm com}(\mathcal{K}\otimes\mathcal{H})$ be the subset of  $\mathcal{J}^{n}(\mathcal{K}\otimes\mathcal{H})$ with its elements corresponding to collections of channels in $\mathfrak{C}_{\rm com}^{n}(\mathcal{H}, \mathcal{K})$. Then, the robustness is rewritten as
 \begin{align}
   \label{robust2}
    R_C(\{\Lambda_x\}) = &\min{s} \notag \\[2pt]
                       & \ \ {\rm s.t.}\  
                       \left\{\frac{J_{\Lambda_x} + sJ_{\Lambda'_x}}{1+s}\right\}\in\mathcal{J}^{n}_{\rm com}(\mathcal{K}\otimes\mathcal{H}) \notag \\[5pt] 
                       &\qquad \ \ 
                       \{J_{\Lambda'_x}\}\in\mathcal{J}^{n}(\mathcal{K}\otimes\mathcal{H}).          
 \end{align}
 
 Since $\mathfrak{C}_{\rm com}^{n}(\mathcal{H}, \mathcal{K})$ is convex and compact, 
 $\mathcal{J}^{n}_{\rm com}(\mathcal{K}\otimes\mathcal{H})$ is also convex and compact by the continuity of the mapping $\Lambda \mapsto J_{\Lambda}$. The compactness of 
 $\mathfrak{C}_{\rm com}^{n}(\mathcal{H}, \mathcal{K})$ follows from the compactness of 
 $\mathfrak{C}(\mathcal{H}, \mathcal{K} \otimes \cdots \otimes \mathcal{K})$ 
 and the continuity of the mapping $\mathfrak{C}(\mathcal{H}, \mathcal{K} \otimes \cdots \otimes \mathcal{K}) \ni \Lambda \mapsto \{\Lambda_x\} \in \mathfrak{C}^{n}(\mathcal{H}, \mathcal{K})$, where $\{\Lambda_x\}$ are defined via Eq.~(\ref{cha}). Hence, we can use the duality theory of optimization problems \cite{boyd2004convex}, and obtain an equivalent formulation of the robustness.
 We present the proof in Appendix \ref{A}.
 \begin{lemma} \label{lemma1} 
 	For any collection of channels $\{\Lambda_x\}\in\mathfrak{C}^{n}           (\mathcal{H}, \mathcal{K})$, the robustness of channel incompatibility $R_C(\{\Lambda_x\})$ can be expressed as follows:
	\begin{align}
	\label{lemma}
	  R_C(\{\Lambda_x\}) = &\max_{\{A_x\}}{\sum_{x}{\rm Tr}[A_x J_{\Lambda_x}]-1} 
	                     \notag \\[2pt]
	         {\rm s.t.}\  &\ A_x\geq 0 \quad \forall x,  \notag \\[2pt] 
	                       &\sum_{x}{\rm Tr}[A_x J_{\Phi_x}]\leq 1 
	                       \quad \forall\{\Phi_x\}\in
	                       \mathfrak{C}^{n}_{\rm com}(\mathcal{H}, \mathcal{K}),
	\end{align} 
	where $\{A_x\}$ are operators on $\mathcal{K}\otimes\mathcal{H}$.
 \end{lemma}

 We consider an entanglement-assisted state discrimination task, that is, a problem to discriminate ensembles of quantum states on the extended Hilbert space $\mathcal{H}\otimes\mathcal{H}$. In this task, Alice has 
 access to channels $\{\Lambda_x\}$ acting on a single system as well as joint measurements $\{\mathbb{M}_x\}$. After receiving the state and the value of $x$, Alice performs the preprocessing of the state by  $\Lambda_x\otimes{\rm id}$ and subsequently the measurement $\mathbb{M}_x$. Then, the average success probability is given by
 \begin{align}
    &P_{\rm succ}(\mathscr{A}, \{\mathbb{M}_x\}, \{\Lambda_x\otimes{\rm id}\}) \notag\\
    &\quad =\sum_{x,i}p(x)p(i|x){\rm Tr}[(\Lambda_x\otimes{\rm id})(\rho_{i|x})M_{i|x}],
 \end{align}
 where $\rho_{i|x}\in\mathcal{S}(\mathcal{H}\otimes\mathcal{H})$ and 
 $\{\mathbb{M}_x\}\in\mathfrak{P}^{n}(\Omega,\mathcal{K}\otimes\mathcal{H})$. 
 
 We are ready to state and prove our main theorem.
 \begin{theorem}
 	\label{thm1}
	For any collection of channels $\{\Lambda_x\}\in\mathfrak{C}^{n}(\mathcal{H}, \mathcal{K})$, 
	the robustness of channel incompatibility $R_C(\{\Lambda_x\})$ quantifies the maximum advantage:
	\begin{align}
	 &\max_{\mathscr{A}, \{\mathbb{M}_x\}}
	 \dfrac{P_{\rm succ}(\mathscr{A}, \{\mathbb{M}_x\}, \{\Lambda_x\otimes{\rm id}\})}{\displaystyle\max_{\{\Phi_x\}\in\mathfrak{C}^{n}_{\rm com}(\mathcal{H}, \mathcal{K})}P_{\rm succ}(\mathscr{A}, \{\mathbb{M}_x\}, \{\Phi_x\otimes{\rm id}\})} \notag \\[2pt]
	 &\qquad = 1+R_C(\{\Lambda_x\}),
	\end{align}
	where the maximization is over all ensembles of states $\mathscr{A}$ and all collections of measurements $\{\mathbb{M}_x\}$.
\end{theorem}
	\begin{proof}
		We write simply $\mathfrak{C}^{n}_{\rm com}(\mathcal{H}, \mathcal{K})$ as $\mathfrak{C}^{n}_{\rm com}$ . By the definition of the robustness, there exists $\{\Phi_x'\}\in\mathfrak{C}^{n}_{\rm com}$ such that
		\begin{align}
		  (\Lambda_x\otimes{\rm id})(\rho)\leq(1+R_C(\{\Lambda_x\}))(\Phi_x'\otimes{\rm id})(\rho)
		\end{align}
		for all $x$ and $\rho\in\mathcal{S}(\mathcal{H}\otimes\mathcal{H})$. Therefore, we have
		\begin{align}
		   &P_{\rm succ}(\mathscr{A}, \{\mathbb{M}_x\}, \{\Lambda_x\otimes{\rm id}\}) \notag \\
		   &=\sum_{x,i}p(x)p(i|x){\rm Tr}[(\Lambda_x\otimes{\rm id})(\rho_{i|x})M_{i|x}] \notag\\
		   &\leq (1+R_C(\{\Lambda_x\}))\sum_{x,i}p(x)p(i|x){\rm Tr}[(\Phi_x'\otimes{\rm id})(\rho_{i|x})M_{i|x}] \notag\\
		   &\leq (1+R_C(\{\Lambda_x\}))\max_{\{\Phi_x\}\in\mathfrak{C}^{n}_{\rm com}}P_{\rm succ}(\mathscr{A}, \{\mathbb{M}_x\}, \{\Phi_x\otimes{\rm id}\}).
		\end{align}
		In order to prove the converse, we choose the measurements $\mathbb{M}_x=\left\{\frac{A_x}{\|A_x\|}, \1-\frac{A_x}{\|A_x\|}\right\}$, where $\{A_x\}$ is an optimal solution in Eq.~(\ref{lemma}), and the ensembles $\mathcal{E}_x=\{p(i|x), \rho_{i|x}\}^{2}_{i=1}$ with probability $p(x)=\frac{\|A_x\|}{\sum_{y}\|A_y\|}$, where $p(1|x)=1$, $\rho_{1|x}=|\Psi_+\rangle\langle\Psi_+|$, $p(2|x)=0$ and $\rho_{2|x}$ is arbitrary for all $x$. For this choice, it holds that 
		\begin{align}
		&\dfrac{P_{\rm succ}
			(\mathscr{A}, \{\mathbb{M}_x\}, \{\Lambda_x\otimes{\rm id}\})}
		{\displaystyle\max_{\{\Phi_x\}\in\mathfrak{C}^{n}_{\rm com}}P_{\rm succ}(\mathscr{A}, \{\mathbb{M}_x\}, \{\Phi_x\otimes{\rm id}\})} \notag\\
		&= \frac{\sum_{x}\frac{\|A_x\|}{\sum_{y}\|A_y\|}{\rm Tr}\left[(\Lambda_x\otimes{\rm id})(|\Psi_+\rangle\langle\Psi_+|)
			\frac{A_x}{\|A_x\|}\right]}
		{{\displaystyle\max_{\{\Phi_x\}\in\mathfrak{C}^{n}_{\rm com}}}
			\sum_{x}\frac{\|A_x\|}{\sum_{y}\|A_y\|}{\rm Tr}\left[(\Phi_x\otimes{\rm id})(|\Psi_+\rangle\langle\Psi_+|)
			\frac{A_x}{\|A_x\|}\right]} \notag\\
		&= \frac{\sum_{x}{\rm Tr}[A_x J_{\Lambda_x}]}{{\displaystyle\max_{\{\Phi_x\}\in\mathfrak{C}^{n}_{\rm com}}}\sum_{x}{\rm Tr}[A_x J_{\Phi_x}]} \notag\\[5pt]
		&\geq 1+R_C(\{\Lambda_x\}),
		\end{align}
		where the inequality follows from Lemma~\ref{lemma1}.
	\end{proof}
From the above theorem, we can see that a collection of channels is incompatible if 
and only if it offers an advantage over all compatible channels as 
preprocessings in a state discrimination task given certain ensembles and measurements. Furthermore, the maximum advantage in the state discrimination task becomes a measure quantifying how incompatible the collection of channels is. As we show later, this theorem includes incompatibility of measurements as a special case.

Unlike in the case of measurement incompatibility, we have considered the entanglement-assisted discrimination task in 
Theorem \ref{thm1}. However, it is not clear that entanglement is truly necessary. We show that entanglement is 
needed in general by giving a simple example. We consider two identity channels $\{{\rm id},{\rm id}\}\in\mathfrak{C}^{2}(\mathcal{H}, \mathcal{H})$. They are incompatible 
and the robustness $R_C(\{{\rm id},{\rm id}\})$ was calculated in \cite{Erkka} as 
\begin{align}
	1+R_C(\{{\rm id},{\rm id}\})=\frac{2d}{d+1}.
\end{align}
On the other hand, the maximum advantage provided by $\{{\rm id},{\rm id}\}$ in the entanglement-unassisted state discrimination task 
can be bounded as (see Appendix \ref{C})
\begin{align}
\label{entan}
	\max_{\mathscr{A}, \{\mathbb{M}_x\}}\dfrac{P_{\rm succ}
		(\mathscr{A}, \{\mathbb{M}_x\}, \{{\rm id},{\rm id}\})}
	{\displaystyle\max_{\{\Phi_x\}\in\mathfrak{C}^{2}_{\rm com}(\mathcal{H}, \mathcal{H})}P_{\rm succ}(\mathscr{A}, \{\mathbb{M}_x\}, \{\Phi_x\})}
	\leq \frac{2(d+1)}{d+3}.
\end{align}
 It is easy to show that the inequality
 \begin{align}
 \frac{2(d+1)}{d+3}<\frac{2d}{d+1}
 \end{align}
 holds for any $d\geq 2$. Therefore, entanglement is necessary in general 
 to characterize the robustness of channel incompatibility as the advantage in a particular state discrimination task.

\subsection{robustness of measurement and channel incompatibility}
In what follows, we write $\mathfrak{P}(\Omega,\mathcal{H})$ and $\mathfrak{C}(\mathcal{H}, \mathcal{K})$ as $\mathfrak{P}$ and $\mathfrak{C}$, respectively 
for simplicity. For any pair of a measurement and a channel $\{\mathbb{M},\Lambda\}\in
\mathfrak{P}\times\mathfrak{C}$, the robustness of measurement and channel incompatibility $R_{MC}(\{\mathbb{M},\Lambda\})$ is defined as
\begin{align}
  \label{pair}
     R_{MC}(\{\mathbb{M},\Lambda\}) = &\min{s} \notag \\[2pt]
                             {\rm s.t.}\  
                             &\left\{\frac{\mathbb{M} + s\mathbb{M}'}{1+s}, \frac{\Lambda + s\Lambda'}{1+s}\right\}
                             \in(\mathfrak{P}\times\mathfrak{C})_{\rm com},
                              \notag \\[5pt] 
                             &\ \{\mathbb{M}',\Lambda'\}\in\mathfrak{P}\times\mathfrak{C}, 
\end{align}
where the minimization is over $s\geq0$ and $\{\mathbb{M}',\Lambda'\}$, and 
$(\mathfrak{P}\times\mathfrak{C})_{\rm com}$ denotes the set of compatible pairs. 
Since $\mathfrak{I}$ is compact, $(\mathfrak{P}\times\mathfrak{C})_{\rm com}$ is also 
compact by the continuity of the mapping 
$\mathfrak{I} \ni \mathcal{I} \mapsto \{\{\mathcal{I}^{*}_i(\1)\},\sum_{i}\mathcal{I}_i\} \in \mathfrak{P}\times\mathfrak{C}$.
Hence, we can use the duality theory of optimization problems, and obtain an equivalent formulation of the robustness. We present the proof in Appendix \ref{B}.
 \begin{lemma} \label{lemma2} 
	For any pair of a measurement and a channel $\{\mathbb{M},\Lambda\}\in
	\mathfrak{P}(\Omega,\mathcal{H})\times\mathfrak{C}(\mathcal{H},\mathcal{K})$, the robustness of measurement and channel incompatibility $R_{MC}(\{\mathbb{M},\Lambda\})$ can be expressed as follows:
	\begin{align}
	\label{lem2}
	&R_{MC}(\{\mathbb{M},\Lambda\}) = \max_{\{A_i\},B}{\ \sum_{i}{\rm Tr}[A_i M_i]+{\rm Tr}[B J_{\Lambda}]-1} 
	\notag \\[2pt]
	&{\rm s.t.}\ \ A_i\geq 0 \quad \forall i, \quad B\geq 0,  \notag \\[2pt] 
	&\quad \sum_{i}{\rm Tr}[A_i N_i]+{\rm Tr}[BJ_{\Phi}]\leq 1 
	\quad \forall\{\mathbb{N},\Phi\}\in
	(\mathfrak{P}\times\mathfrak{C})_{\rm com}.
	\end{align} 
	where $\{A_i\}$ are operators on $\mathcal{H}$ and $B$ is an operator on $\mathcal{K}\otimes\mathcal{H}$.
\end{lemma}

We consider the state discrimination that Bob prepares two different ensembles 
$\mathscr{A}=\{p(x),\mathcal{E}_x\}_{x=1}^{2}$ on the extended Hilbert space 
$\mathcal{H}\otimes\mathcal{H}$ and Alice has access to a pair of a measurement and a channel $\{\mathbb{M},\Lambda\}\in\mathfrak{P}(\Omega,\mathcal{H})\times\mathfrak{C}(\mathcal{H}, \mathcal{K})$ as well as a single measurement $\mathbb{L}\in\mathfrak{P}(\Omega,\mathcal{K}\otimes\mathcal{H})$. If she receives the value $x=1$, she performs the measurement $\mathbb{M}\otimes\1$, and if she receives the value $x=2$, she performs 
the preprocessing $\Lambda\otimes{\rm id}$ and subsequently the measurement $\mathbb{L}$.
The average success probability is then given by 
\begin{align}
	&P_{\rm succ}(\mathscr{A}, \mathbb{L}, \{\mathbb{M}\otimes\1,\Lambda\otimes{\rm id}\}) \notag\\
	&\qquad  =\sum_{i}p(1)p(i|1){\rm Tr}[\rho_{i|1}M_i\otimes\1] \notag \\
	&\quad\qquad +\sum_{i}p(2)p(i|2){\rm Tr}[(\Lambda\otimes{\rm id})(\rho_{i|2})L_i].
\end{align}
 
 In this setting, we show that a similar relation to Eq.~(\ref{measurement}) and 
 Theorem~\ref{thm1} holds for incompatibility of measurement and channel.
 \begin{theorem}
 	\label{thm2}
	For any pair of a measurement and a channel $\{\mathbb{M},\Lambda\}\in
	\mathfrak{P}(\Omega,\mathcal{H})\times\mathfrak{C}(\mathcal{H},\mathcal{K})$, 
	the robustness of measurement and channel incompatibility quantifies the maximum advantage:{\rm
	\begin{align}
	  &\max_{\mathscr{A}, \mathbb{L}}
	  \dfrac{P_{\rm succ}(\mathscr{A}, \mathbb{L}, \{\mathbb{M}\otimes\1,\Lambda\otimes{\rm id}\})}
	  {\displaystyle\max_{\{\mathbb{N},\Phi\}\in
	  		(\mathfrak{P}\times\mathfrak{C})_{\rm com}}P_{\rm succ}(\mathscr{A}, \mathbb{L}, \{\mathbb{N}\otimes\1,\Phi\otimes{\rm id}\})} \notag \\[4pt]
	  &\qquad = 1+R_{MC}(\{\mathbb{M},\Lambda\}), 
	\end{align}}
    where the maximization is over all ensembles of states $\mathscr{A}$ and all measurements $\mathbb{L}$.
    \begin{proof}
        In the same way as Theorem~\ref{thm1}, it holds that 
    	\begin{align}
    	   &P_{\rm succ}(\mathscr{A}, \mathbb{L}, \{\mathbb{M}\otimes\1,\Lambda\otimes{\rm id}\}) \notag \\
    	   &\leq (1+R_{MC}(\{\mathbb{M},\Lambda\})) \notag\\
    	   &\quad \times\max_{\{\mathbb{N},\Phi\}\in(\mathfrak{P}\times\mathfrak{C})_{\rm com}}P_{\rm succ}(\mathscr{A}, \mathbb{L}, \{\mathbb{N}\otimes\1,\Phi\otimes{\rm id}\}).
    	\end{align}
    	In order to prove the converse, we choose the ensembles $\mathscr{A}$ and the measurement $\mathbb{L}$ as follows:
    	\begin{align}
    	   &\rho_{i|1}=\frac{A_i}{{\rm Tr}[A_i]}\otimes\sigma_i,\quad p(i|1)=\frac{{\rm Tr}[A_i]}{\sum_{i}{\rm Tr}[A_i]} \notag \\[1pt]
    	   &\rho_{1|2}=|\Psi_+\rangle\langle\Psi_+|,\quad \rho_{i|2}:{\rm arbitrary}\ (i=2,\cdots, o) \notag\\[1pt]
    	   &p(1|2)=1, \quad p(i|2)=0\ (i=2,\cdots,o) \notag\\[1pt]
    	   &p(1)=\frac{\sum_{i}{\rm Tr}[A_i]}{\sum_{i}{\rm Tr}[A_i]+\|B\|},\quad 
    	   p(2)=\frac{\|B\|}{\sum_{i}{\rm Tr}[A_i]+\|B\|} \notag \\[1pt]
    	   &L_1=\frac{B}{\|B\|},\ L_2=\1-\frac{B}{\|B\|},\ L_i=0\ (i=3,\cdots,o),
    	\end{align}
    	where $\{\{A_i\},B\}$ is the optimal solution for Eq.~(\ref{lem2}), and 
    	$\sigma_i$ is an arbitrary state for all $i$.
    	For this choice, it holds that 
    	\begin{align}
    	   &\dfrac{P_{\rm succ}(\mathscr{A}, \mathbb{L}, \{\mathbb{M}\otimes\1,\Lambda\otimes{\rm id}\})}
    	   {\displaystyle\max_{\{\mathbb{N},\Phi\}\in
    			(\mathfrak{P}\times\mathfrak{C})_{\rm com}}P_{\rm succ}(\mathscr{A}, \mathbb{L}, \{\mathbb{N}\otimes\1,\Phi\otimes{\rm id}\})} \notag\\
    	   &\quad=\dfrac{\sum_{i}{\rm Tr}[A_i M_i]+{\rm Tr}[B J_{\Lambda}]}
    	   {\displaystyle\max_{\{\mathbb{N},\Phi\}\in
    			(\mathfrak{P}\times\mathfrak{C})_{\rm com}}\sum_{i}{\rm Tr}[A_i N_i]+{\rm Tr}[B J_{\Phi}]} \notag\\
    	   &\quad\geq 1+R_{MC}(\{\mathbb{M},\Lambda\}),
    	\end{align}
    	where the inequality follows from Lemma~\ref{lemma2}.
    \end{proof}
\end{theorem}

Thus, incompatibility of measurement and channel also has an operational 
interpretation in terms of state discrimination, that is, every incompatible pair of 
a measurement and a channel provides an advantage in a state discrimination task with two different ensembles of states. 

\subsection{quantum-to-classical channels and robustness of incompatibility}
Let us finally show that the robustness of measurement incompatibility and the robustness
of measurement and channel incompatibility are obtained as special cases from the 
robustness of channel incompatibility. We can regard a measurement $\mathbb{M}$ as a special kind of channel via a quantum-to-classical channel $\Gamma_{\mathbb{M}}$ defined as　
\begin{align}
	\Gamma_{\mathbb{M}}(\rho) = \sum_{i}{\rm Tr}[\rho M_i]|i\rangle\langle i|,
\end{align}
where $\{\ket{i}\}$ is an orthonormal basis for the output Hilbert space $\mathcal{K}$. 

It can be shown that measurements $\{\mathbb{M}_x\}$ are compatible if and only if 
the corresponding quantum-to-classical channels $\{\Gamma_{\mathbb{M}_x}\}$ 
are compatible \cite{Teiko2017}. Hence, we expect that the degree of incompatibility for 
$\{\mathbb{M}_x\}$ is also equivalent to the degree of incompatibility for 
$\{\Gamma_{\mathbb{M}_x}\}$.
We prove that this is the case.
\begin{proposition}
	\label{prop1}
	For any collection of measurements $\{\mathbb{M}_x\}\in\mathfrak{P}^{n}(\Omega,\mathcal{H})$, the robustness of 
	measurement incompatibility for $\{\mathbb{M}_x\}$ is equal to the robustness of 
	channel incompatibility for the corresponding quantum-to-classical channels 
	$\{\Gamma_{\mathbb{M}_x}\}$:
	\begin{align}
		R_M(\{\mathbb{M}_x\}) = R_C(\{\Gamma_{\mathbb{M}_x}\}).
	\end{align}
	\begin{proof}
		Since the mapping $\mathbb{M}_x \mapsto \Gamma_{\mathbb{M}_x}$ is affine, 
		the compatibility of $\left\{\frac{\mathbb{M}_x + s\mathbb{M}_x'}{1+s}\right\}$ leads to the compatibility of $\left\{\frac{\Gamma_{\mathbb{M}_x} + s\Gamma_{\mathbb{M}_x'}}{1+s}\right\}$. Hence, we have 
		\begin{align}
			R_M(\{\mathbb{M}_x\}) \geq R_C(\{\Gamma_{\mathbb{M}_x}\}).
		\end{align}
		In order to prove the converse, we suppose 
		$\{\Phi_x\}=\left\{\frac{\Gamma_{\mathbb{M}_x} + s\Lambda_x}{1+s}\right\}$
		are compatible. Then, it holds that 
		\begin{align}
			{\rm Tr}[\Gamma_{\mathbb{M}_x}(\rho)N_{i|x}] \leq 
			(1+s){\rm Tr}[\Phi_x(\rho)N_{i|x}] 
		\end{align} 
		for all $\rho$, $\{\mathbb{N}_x\}$, $x$ and $i$. 
		By taking the measurements $\{\mathbb{N}_x\}$ as $N_{i|x}=|i\rangle\langle i|$ 
		for all $i$ and $x$, we have 
		\begin{align}
			{\rm Tr}[\rho M_{i|x}] \leq 
			(1+s){\rm Tr}[\rho\Phi_x^{*}(|i\rangle\langle i|)] \quad \forall\ i,x,\rho,
		\end{align}
		that is, 
		\begin{align}
			M_{i|x} \leq (1+s)\Phi_x^{*}(|i\rangle\langle i|)\quad \forall\ i,x.
		\end{align}
		The measurements $\{\Phi_x^{*}(\mathbb{N}_x)\}$ are compatible due to 
		the compatibility of $\{\Phi_x\}$ \cite{Teiko2017}. 
		By defining measurements $\{\mathbb{M}_x'\}$ as 
		\begin{align}
		M_{i|x}'=\frac{(1+s)\Phi_{x}^{*}(|i\rangle\langle i|)-M_{i|x}}{s}\quad \forall\ i,x,
		\end{align}
		$\{\Phi_{x}^{*}(\mathbb{N}_x)\}$ can be written as 
		$\left\{\frac{\mathbb{M}_x+s\mathbb{M}_x'}{1+s}\right\}$.
		Therefore, we get 
		\begin{align}
			R_M(\{\mathbb{M}_x\}) \leq R_C(\{\Gamma_{\mathbb{M}_x}\}).
		\end{align}
	\end{proof}
\end{proposition}
Hence, Theorem~\ref{thm1} can not only become a witness of measurement incompatibility but also quantify precisely the robustness of measurement incompatibility via quantum-to-classical channels.

It can also be shown that a measurement $\mathbb{M}$ and a channel $\Lambda$ are 
compatible if and only if the corresponding quantum-to-classical channel $\Gamma_{\mathbb{M}}$ and $\Lambda$ are compatible \cite{Teiko2017}.  
\begin{proposition}
	\label{prop2}
	For any pair of a measurement and a channel 
	$\{\mathbb{M}, \Lambda\} \in \mathfrak{P}(\Omega,\mathcal{H})\times\mathfrak{C}(\mathcal{H},\mathcal{K})$, the robustness of measurement and 
	channel incompatibility for $\{\mathbb{M}, \Lambda\}$ is equal to the robustness 
	of channel incompatibility for $\{\Gamma_{\mathbb{M}},\Lambda\}$:
	\begin{align}
		R_{MC}(\{\mathbb{M}, \Lambda\})=R_C(\{\Gamma_{\mathbb{M}},\Lambda\}).
	\end{align}
	\begin{proof}
		In the same way as Proposition~\ref{prop1}, we have 
		\begin{align}
			R_{MC}(\{\mathbb{M}, \Lambda\})\geq R_C(\{\Gamma_{\mathbb{M}},\Lambda\}).
		\end{align}
		We suppose $\{\Phi_1,\Phi_2\} = \left\{\frac{\Gamma_{\mathbb{M}} + s\Phi'}{1+s},\frac{\Lambda + s\Lambda'}{1+s}\right\}$ are compatible. Then, as with Proposition~\ref{prop1}, 
		it holds that 
		\begin{align}
			M_i \leq (1+s)\Phi_{1}^{*}(|i\rangle\langle i|)\quad \forall\ i.
		\end{align}
		Let $\Phi$ be a joint channel of $\Phi_1$ and $\Phi_2$. We define an instrument $\{\mathcal{I}_i\}$ as 
		\begin{align}
			\mathcal{I}_{i}^{*}(A)=\Phi^{*}(|i\rangle\langle i|\otimes A) \quad \forall\ i.
		\end{align}
		Then 
		\begin{align}
			\mathcal{I}_{i}^{*}(\1)=\Phi_{1}^{*}(|i\rangle\langle i|),\quad 
			\sum_{i}\mathcal{I}_{i}^{*}(A)=\Phi_{2}^{*}(A)\quad \forall\ i,A,
		\end{align}
		hence, $\{\Phi_{1}^{*}(|i\rangle\langle i|) \}$ and $\Phi_2$ are compatible. 
		By defining a measurement $\mathbb{M}'$ as 
		\begin{align}
			M_i'=\frac{(1+s)\Phi_{1}^{*}(|i\rangle\langle i|)-M_i}{s}\quad \forall\ i,
		\end{align}
		$\{\Phi_{1}^{*}(|i\rangle\langle i|)\}$ can be written as $\frac{\mathbb{M}+s\mathbb{M}'}{1+s}$. Therefore, we get
		\begin{align}
			R_{MC}(\{\mathbb{M}, \Lambda\})\leq R_C(\{\Gamma_{\mathbb{M}},\Lambda\}).
		\end{align}
	\end{proof}
\end{proposition}
According to Proposition~\ref{prop1} and Proposition~\ref{prop2}, we can see that Theorem~\ref{thm1} 
includes incompatibility of measurements and incompatibility of measurement 
and channel as special cases. However, Theorem \ref{thm2} is not  a mere corollary of Theorem \ref{thm1} and Proposition \ref{prop2} because Theorem \ref{thm1} and Theorem \ref{thm2} are different in the way to maximize the denominators.
\section{conclusions}
\label{conc}
In this paper, we have shown that the robustness of channel incompatibility exactly quantifies the maximum advantage in a state discrimination task with multiple ensembles of states as with the robustness of measurement incompatibility.
This shows that every collection of incompatible channels can be more useful than compatible ones. 
In addition, unlike in the case of measurement incompatibility, it has been proven by an example that entanglement is essential. We have also proven that a similar relation holds for the robustness of measurement and 
channel incompatibility.
Therefore, it has been shown that every kind of quantum incompatibility is operationally characterized by quantum state discrimination. 
This enables us to see quantum incompatibility as a kind of useful resource.

\section*{Note added}
Recently, we have become aware of independent works related to our Theorem~\ref{thm1} by R. Uola {\em et al.} \cite{RoopeChan} and by C. Carmeli {\em et al.} \cite{CarmeliWit}. 
In both works, the authors show that a collection of channels is incompatible 
if and only if it offers an advantage over all compatible ones in quantum information games. In particular, in Ref.~\cite{RoopeChan} the authors consider 
input-output games that do not need entanglement but differ from our simple state discrimination games.
They prove that the robustness of channel incompatibility quantifies the maximum advantage in those games. Although in Ref.~\cite{CarmeliWit} the authors do not provide 
a quantitative relation between the robustness and the advantage, 
it gives a necessary and sufficient condition for channel incompatibility 
in terms of state discrimination without entanglement.
On the other hand, as we have shown, entanglement is necessary for characterizing the robustness in terms of state discrimination.    

\appendix
\section{proof of Lemma 1}
\label{A}
	From Eq.~(\ref{robust2}), 
	\begin{align}
	R_C(\{\Lambda_x\}) = &\min{s} \notag \\
	& \ \ {\rm s.t.}\ \  J_{\Lambda_x} \leq (1+s)J_{\Phi_x} 
	\quad \forall x,  \notag \\
	&\qquad \ \ \{J_{\Phi_x}\}\in\mathcal{J}^{n}_{\rm com}(\mathcal{K}\otimes\mathcal{H}). 
	\end{align}
	By defining $\{\tilde{J}_{\Phi_x}\}$ with $\tilde{J}_{\Phi_x}:=(1+s)J_{\Phi_x}$ and using ${\rm Tr}[J_{\Phi_x}]=1$, $R_C(\{\Lambda_x\})$ can be equivalently written as
	\begin{align}
	\label{primal}
	R_C(\{\Lambda_x\}) = &\min{\sum_{x}
		\frac{{\rm Tr}[\tilde{J}_{\Phi_x}]}{n}-1} \notag \\
	& \ \ {\rm s.t.}\ \  J_{\Lambda_x} \leq \tilde{J}_{\Phi_x} 
	\quad \forall x,  \notag \\
	&\qquad \ \ \{\tilde{J}_{\Phi_x}\}\in {\rm cone} \bigl(\mathcal{J}^{n}_{\rm com}
	(\mathcal{K}\otimes\mathcal{H})\bigr),	 
	\end{align}
	where ${\rm cone} \bigl(\mathcal{J}^{n}_{\rm com}(\mathcal{K}\otimes\mathcal{H})\bigr)$ is the cone generated by $\mathcal{J}^{n}_{\rm com}(\mathcal{K}\otimes\mathcal{H})$. Since $\mathcal{J}^{n}_{\rm com}(\mathcal{K}\otimes\mathcal{H})$ is compact, ${\rm cone} \bigl(\mathcal{J}^{n}_{\rm com}(\mathcal{K}\otimes\mathcal{H})\bigr)$ is closed. Therefore, we can use the conic programming. Define the Lagrangian associated to this optimization problem
	\begin{align}
	\mathcal{L} &= \sum_{x}\frac{{\rm Tr}[\tilde{J}_{\Phi_x}]}{n}-1 \notag\\
	&\quad -\sum_{x}{\rm Tr}[A_x(\tilde{J}_{\Phi_x}-J_{\Lambda_x})]
	-\sum_{x}{\rm Tr}[B_x\tilde{J}_{\Phi_x}] \notag \\
	&=\sum_{x}{\rm Tr}[A_x J_{\Lambda_x}]-1+\sum_{x}{\rm Tr}\left[\left(\frac{\1}{n}-A_x-B_x\right)\tilde{J}_{\Phi_x}\right],
	\end{align}
	where $A_x$ is positive semidefinite for all $x$ and $\{B_x\}$ is the element of the dual cone of ${\rm cone} \bigl(\mathcal{J}^{n}_{\rm com}(\mathcal{K}\otimes\mathcal{H})\bigr)$. The Lagrangian becomes independent of the primal variables if we restrict to dual variables that satisfy $A_x+B_x=\1/n$ for all $x$. Therefore, the dual form of Eq.~(\ref{primal}) becomes
	\begin{align}
	\label{lem}
	&\max_{\{A_x\}}{\sum_{x}{\rm Tr}[A_x J_{\Lambda_x}]-1} \notag \\[2pt]
	& \ \ {\rm s.t.}\ \  A_x\geq 0 \quad \forall x,  \notag \\ 
	&\qquad \ \sum_{x}{\rm Tr}[A_x J_{\Phi_x}]\leq 1 \quad \forall\{\Phi_x\}\in\mathfrak{C}^{n}_{\rm com}(\mathcal{H}, \mathcal{K}).
	\end{align}
	The second constraint of Eq.~(\ref{lem}) follows from $\sum_{x}{\rm Tr}[B_x J_{\Phi_x}]\geq 0 $ for all $\{J_{\Phi_x}\}\in \mathcal{J}^{n}_{\rm com}(\mathcal{K}\otimes\mathcal{H})$. The optimal values of the primal and the dual form coincide if strong duality holds. This is satisfied if Eq.~(\ref{primal}) is finite and satisfies Slater's condition. It is clear that Eq.~(\ref{primal}) is finite 
	because the robustness is finite. 
	We define $\{\Phi_x\}\in\mathfrak{C}^{n}(\mathcal{H}, \mathcal{K})$ as
	$\Phi_x(\cdot) = \1_{\mathcal{K}}/d'$, where $d'$ is the dimension of $\mathcal{K}$.
	$\{\Phi_x\}$ are obviously compatible. The corresponding Choi matrices are given by
	$J_{\Phi_x}=\1/dd'$. 
	Therefore, $\{J_{\Phi_x}\}$ can be multiplied by a sufficiently large positive number
	to be the strictly feasible point of the first constraint of Eq.~(\ref{primal}).
	Moreover, $\{J_{\Phi_x}\}$ is the interior point of 
	${\rm cone} \bigl(\mathcal{J}^{n}_{\rm com}(\mathcal{K}\otimes\mathcal{H})\bigr)$.
	Hence, the optimal values of both problems coincide.
	\\
	
\section{proof of Lemma 2}
\label{B}
	The proof is the same as Lemma~\ref{lemma1}. That is, by using the conic programming, we can obtain the dual problem Eq.~(\ref{lem2}). To fulfill Slater's condition, we may 
	take the point $\mathbb{N}=\{p(i)\1\}$ and $\Phi(\cdot)=\1_{\mathcal{K}}/d'$, where $p(i)$ is the 
	probability distribution. They are obviously compatible.

\section{proof of Eq.~(\ref{entan})}
\label{C}
In order to evaluate the maximum advantage in the entanglement-unassisted state discrimination task, 
we introduce the optimal cloning channel $\mathcal{C}\in\mathfrak{C}(\mathcal{H},\mathcal{H}\otimes\mathcal{H})$ of the form \cite{keyl1999optimal}
\begin{align}
\mathcal{C}(\rho)=\frac{2}{d+1}S(\rho\otimes\1)S,
\end{align}
where $S$ is the projection from $\mathcal{H}\otimes\mathcal{H}$ to the symmetric subspace of $\mathcal{H}\otimes\mathcal{H}$.
The two channels $\{\mathcal{C}_1,\mathcal{C}_2\}$ obtained as the corresponding marginals of $\mathcal{C}$ read 
\begin{align}
\mathcal{C}_1(\rho)=\mathcal{C}_2(\rho)=c(d)\rho+(1-c(d))\frac{\1}{d},
\end{align}
where the number $c(d)$ is given by
\begin{align}
c(d)=\frac{d+2}{2(d+1)}.
\end{align}
Since the two channels $\{\mathcal{C}_1,\mathcal{C}_2\}$ are compatible by their construction, we have
\begin{widetext}
	\begin{align}
	\max_{\mathscr{A}, \{\mathbb{M}_x\}}\dfrac{P_{\rm succ}
		(\mathscr{A}, \{\mathbb{M}_x\}, \{{\rm id},{\rm id}\})}
	{\displaystyle\max_{\{\Phi_x\}\in\mathfrak{C}^{2}_{\rm com}(\mathcal{H}, \mathcal{H})}P_{\rm succ}(\mathscr{A}, \{\mathbb{M}_x\}, \{\Phi_x\})}
	&\leq \max_{\mathscr{A}, \{\mathbb{M}_x\}}\dfrac{P_{\rm succ}
		(\mathscr{A}, \{\mathbb{M}_x\}, \{{\rm id},{\rm id}\})}
	{P_{\rm succ}(\mathscr{A}, \{\mathbb{M}_x\}, \{\mathcal{C}_1,\mathcal{C}_2\})} \notag\\
	&= \max_{\mathscr{A}, \{\mathbb{M}_x\}}\frac{\displaystyle\sum_{x=1}^{2}\sum_{i=1}^{o}p(x)p(i|x){\rm Tr}[\rho_{i|x}M_{i|x}]}
	{\displaystyle\sum_{x=1}^{2}\sum_{i=1}^{o}p(x)p(i|x)\left(c(d){\rm Tr}[\rho_{i|x}M_{i|x}]+
		\frac{1-c(d)}{d}{\rm Tr}[M_{i|x}]\right)} \notag\\
	&= \max_{\{p(x), p(i|x)\}, \{\mathbb{M}_x\}}\frac{\displaystyle\sum_{x=1}^{2}\sum_{i=1}^{o}p(x)p(i|x)\|M_{i|x}\|}
	{\displaystyle\sum_{x=1}^{2}\sum_{i=1}^{o}p(x)p(i|x)\left(c(d)\|M_{i|x}\|+
		\frac{1-c(d)}{d}{\rm Tr}[M_{i|x}]\right)} \notag\\
	&\leq  \max_{\{p(x), p(i|x)\}, \{\mathbb{M}_x\}}\frac{\displaystyle\sum_{x=1}^{2}\sum_{i=1}^{o}p(x)p(i|x)\|M_{i|x}\|}
	{\left(c(d)+\frac{1-c(d)}{d}\right)\displaystyle\sum_{x=1}^{2}\sum_{i=1}^{o}p(x)p(i|x)\|M_{i|x}\|} \notag\\
	&=\frac{1}{c(d)+\frac{1-c(d)}{d}}=\frac{2(d+1)}{d+3},
	\end{align} 
\end{widetext}
where the third line is because the function $f(x)=\frac{x}{ax+b}\ (a,b>0)$ is monotonically increasing in $x\geq 0$ 
and the fourth line is due to the fact $\|M_{i|x}\|\leq{\rm Tr}[M_{i|x}]$.


\end{document}